\documentclass[letterpaper, 10 pt, conference]{ieeeconf}  
\IEEEoverridecommandlockouts 
\pdfoutput=1

\usepackage{amsmath,amssymb,color,cite}
\usepackage{graphicx}

\usepackage{latexsym}
\usepackage{algorithm, algorithmic}
\usepackage{verbatim}
\usepackage{cite}
\usepackage[english]{babel}
\usepackage{balance}
\usepackage[bookmarks=false]{hyperref}
\usepackage{booktabs}
\hypersetup{colorlinks=,linkcolor=,urlcolor=}

\let\labelindent\relax
\usepackage{amsmath,amsthm}
\usepackage{enumitem}

\usepackage{xcolor}

\usepackage{pgf}
\usepackage{tikz}
\usetikzlibrary{tikzmark}
\usepackage{arydshln}

\usepackage{amsfonts}
\usepackage{amssymb}
\usepackage{amsbsy}
\usepackage{bm}
\usepackage{mathtools}
\usepackage[font=scriptsize]{subcaption}
\usepackage[font=scriptsize]{caption}
\theoremstyle{definition}

\theoremstyle{plain}
\newtheorem{assumption}{Assumption}
\newtheorem{lemma}{Lemma}
\newtheorem{theorem}{Theorem}

\newtheorem{Remark}{Remark}[section]
\newcommand{\bRemark}{ \begin{Remark} \rm }
\newcommand{\eRemark}{ \end{Remark} }

\newcommand{\Scal}{\mathcal{S}}
\newcommand{\Mcal}{\mathcal{M}}
\newcommand{\Acal}{\mathcal{A}}
\newcommand{\Hcal}{\mathcal{H}}
\newcommand{\EE}{\mathbb{E}}
\newcommand{\PP}{\mathbb{P}}
\newcommand{\pistar}{{\pi^*}}
\newcommand{\Vstar}{V^*}

\newcommand{\hthresh}{{\bar{h}}}
\newcommand{\pithresh}{\pi_{t,\bar{h}}}

\newcommand{\Fcal}{\mathcal{F}}
\newcommand{\thetahat}{\hat\theta}
\newcommand{\hhat}{\hat h}

\newcommand{\hstar}{h^\star}

\newcommand{\hhigh}{h_{\text{high}}}
\newcommand{\hlow}{h_{\text{low}}}

\title{\LARGE \bf Safe Learning of Adaptive Control Policies for Remote Patient Monitoring
}

\author{Ramanan Tamizholi$^{1}$, Siddharth Chandak$^{2}$, Isha Thapa$^{3}$, Nicholas Bambos$^{2,3}$ and David Scheinker$^{4}$
\thanks{$^{1}$ Indian Institute of Science (IISc), Bangalore, India.}
\thanks{$^{2}$ Department of Electrical Engineering, Stanford University, USA.}
\thanks{$^{3}$ Department of Management Science \& Engineering, Stanford University, USA.}
\thanks{$^{4}$ School of Medicine, Stanford University, USA.}
\thanks{Emails: ramanant@iisc.ac.in, \{chandaks, ishadt, bambos, dscheink\}@stanford.edu}
}

\begin{document}
\maketitle

\begin{abstract} 
Remote Patient Monitoring (RPM) enables continuous observation of patients in their daily environments, improving both health outcomes and quality of life. A key challenge in RPM is determining the optimal monitoring intensity, while balancing patient safety and monitoring costs. This problem is further complicated when system parameters, such as transition probabilities and costs, are initially unknown. We develop a learning-based control framework that estimates these parameters and adapts the monitoring policy in real time. The proposed approach is an online model-based reinforcement learning algorithm tailored to RPM, with patient safety explicitly prioritized during exploration. We provide theoretical guarantees on safety and convergence to the optimal policy. Simulation results show that the algorithm converges to the optimal threshold-based policy, maintains low treatment costs, and reduces the risk of patients reaching critical health states.
\end{abstract}


\section{Introduction}
Remote Patient Monitoring (RPM) is becoming an increasingly integral part of modern healthcare, enabling continuous observation of patients within their daily environments and enhancing both their quality of life and healthcare outcomes \cite{Farias, Malasinghe, Zinzuwadia}. Advances in wearable medical devices, such as continuous glucose monitors \cite{Lee} and smartwatches with vital sign monitoring capabilities \cite{Masoumian} among others have facilitated the collection and transmission of extensive health data for real-time analysis and timely medical interventions.

A significant challenge in RPM is determining the optimal intensity of patient monitoring. In clinical settings, adjusting monitoring intensity may involve increasing the frequency of data review, initiating outreach to patients, or escalating to in-person visits, significantly impacting provider workload. From the patient's perspective, while intensive monitoring can lead to earlier detection of health issues and prompt responses, it may also cause treatment burden or fatigue due to its invasive nature \cite{heckman}. Conversely, less intensive monitoring is less intrusive to the patient, but might not provide sufficient patient support. This makes the decision of when and how intensively to monitor a key operational and clinical consideration, one that requires balancing health outcomes and monitoring costs by dynamically adapting the monitoring intensity to the patient’s health state.

Recent works  \cite{chandak1, chandak2} addressed this trade-off by modeling the RPM service as a Markov Decision Process (MDP), and showed that \textit{threshold policies}, where patients switch to intensive monitoring once their health falls below a certain threshold are optimal. These analyses characterized how the optimal policy depends on key parameters such as health improvement probabilities, monitoring options, and invasiveness costs, enabling the framework to be adapted to specific medical contexts and individualized patient needs. However, these works assumed that the transition probabilities and patient costs in the MDP are known a priori to clinicians. 

In practice, parameters may be unknown at the time of service deployment. Newer interventions may lack historical data, while established programs may encounter changing clinical guidelines for how health metrics are defined and prioritized \cite{thapa-rdt}. RPM systems naturally generate high-frequency patient data streams, creating the opportunity to adapt monitoring strategies in response to observed patient outcomes during deployment. However, doing so safely requires balancing the need to explore untested actions with the goal of consistently delivering high-quality care.

To address this challenge, we formulate the problem as a Reinforcement Learning (RL) task \cite{Sutton}. Most works applying RL in healthcare focus on optimizing treatment strategies, such as medication timing and dosing \cite{Watts2020, Jafar2024, Luckett2020}. In contrast, we address the monitoring control problem: learning when to intensively observe patients. Rather than automating treatment decisions, we develop an interpretable support tool that helps clinicians allocate attention effectively.

We propose an algorithm that simultaneously learns the system parameters and adapts the monitoring policy in real time. Since learning occurs while the RPM service is deployed on patients, exploration must be carried out without compromising patient safety or imposing excessive monitoring costs. To address this, we introduce a simple randomized exploration scheme tailored to the tiered RPM setting. By occasionally adjusting the deployed threshold toward clinically meaningful health levels, the algorithm promotes exploration while preserving patient safety. We show that, with appropriate initialization, the algorithm satisfies safety guarantees during learning and that the resulting policy converges to the optimal threshold policy under mild assumptions. Simulations further demonstrate that the learned policy converges to the optimal threshold control while maintaining low costs and reducing the risk of patients reaching critical health states.

Our framework offers a practical approach to adaptive monitoring in patient care, addressing a core yet understudied dimension of clinical decision-making: when and how intensively to observe patients. By learning system dynamics during deployment, our method enables real-time adaptation to evolving patient data, institutional constraints, and clinical guidelines without assuming prior knowledge. The resulting threshold policies are transparent and interpretable, providing clinicians with actionable guidance while preserving human oversight. These features make our approach a promising step toward more flexible and deployable RPM tools.

\section{The Remote Patient Monitoring Model} 
\label{sec:model}

Consider a patient who can be in a {\bf health state} $h_t\in\{0,1,2,3,...,H\}$ in each service time period $t\in\{0,1,2,3,...\}$. The RPM service places the patient in a monitoring/intervention state $m_t\in\Mcal=\{o,i\}$ in each time period $t$, abbreviated to {\bf monitoring state}, where $o$ denotes ordinary monitoring and $i$ intensive monitoring. Here, intensive monitoring includes the downstream clinical responses enabled by increased observation, and can therefore lead to improved patient health outcomes. Thus, one can view the 
monitoring-patient joint state $s_t\coloneqq(m_t, h_t)\in\Mcal\times\Hcal\eqqcolon\Scal$ as the system or service state at time $t$.

A higher patient health state $h\in\{0,1,2,3,...,H\}$ corresponds to the patient having better health. In particular, the lowest health state $0$ is {\bf critical} in the sense that, when the patient drifts into that state under any monitoring state $i$ or $o$, they go beyond the scope of the current service model and the service evolution stops. At that point other emergency and/or more severe medical interventions are required.

An advantage of our model is that the costs can be interpreted from multiple perspectives. But in this section, we first take the patient's treatment burden \cite{heckman} point of view. Under ordinary monitoring, the patient incurs a constant cost $C_o\ge0$ at any state $(o,h)$ with $h\in\{1,2,\ldots,H\}$. Correspondingly, under intensive monitoring, the patient incurs a constant cost $C_i\ge0$ at any state $(i,h)$. These costs reflect the invasiveness of the monitoring process to the patient's everyday lifestyle and quality of life. Of special interest are the critical health states, where this model ceases to apply. In such states, a significant cost of $C_c$ is incurred.

\subsection{Markov Decision Process}
We model the system as a Markov Decision Process (MDP). Such models are commonly used for medical decision making \cite{steimle, alagoz}. At the beginning of every time period $t$, the service decides whether to change the monitoring intensity. Formally, the decision/action space is $\Acal=\{o,i\}$ and each (state, action) pair is associated with a cost given by the function $c:\Scal\times\Acal\mapsto\mathbb{R}^+$. The transition probabilities are given by $p(s'|s,a)$ where $s',s\in\Scal$ and $a\in\Acal$. The cost functions and transition probabilities are defined as follows, and have been illustrated in Figure \ref{fig:rpm_model}.

\noindent\textbf{1. At health state $\mathbf{h=0}$:} 
\begin{quote}
    No action is taken and these states are absorbing states. A cost of $C_c$ is incurred.
\end{quote}

\noindent \textbf{2. At health states $\mathbf{1\leq h\leq H}$}:
\begin{enumerate}[label=\alph*)]
\item {\em Ordinary Monitoring ($m=o$), no Switching ($a=o$):} 
No monitoring change. The transition from $(o,h)$ is:
\begin{enumerate}[label=\roman*)]
\item 
$(o,h)\xrightarrow{a=o} (o,\min\{h+1,H\})$ with prob. $\lambda_o$
\item
$(o,h)\xrightarrow{a=o} (o,h-1)$ with prob. $\mu_o=1-\lambda_o$,
\end{enumerate}
and a cost $C_o$ is incurred. Note that $\min\{h+1,H\}$ above is used to account for $(o,H)\to (o,H)$ with prob. $\lambda_o$ since $H$ is the highest health state. 

\item {\em Intensive Monitoring ($m=i$), no Switching ($a=i$):}
No monitoring change. The transition from $(i,h)$ is:
\begin{enumerate}[label=\roman*)]
\item 
$(i,h)\xrightarrow{a=i}(i,\min\{h+1,H\})$ with prob. $\lambda_i$
\item
$(i,h)\xrightarrow{a=i}(i,h-1)$ with prob. $\mu_i=1-\lambda_i$,
\end{enumerate}
and a cost $c\big((i,h),i\big)=C_i$ is incurred.

\item {\em Intensive Monitoring ($m=i$) with Switching ($a=o$):}
Induces a switch to ordinary monitoring. The next state with respective transition probabilities, and the cost incurred is same as part (a): ordinary monitoring ($m=o$) with no switching $(a=o)$.

\item {\em Ordinary Monitoring (m=o), with Switching (a=i):}\\
Induces a switch to intensive monitoring. The next state with respective transition probabilities, and the cost incurred is same as part (b): intensive monitoring ($m=i$) with no switching $(a=i)$.
\end{enumerate}

The decision to switch (or not) is made in the beginning of the timestep. As a result, the transition probabilities and costs are decided solely by the control and the next monitoring state. We make the following \textit{natural} assumptions.
\begin{assumption}
\begin{enumerate}[label=\alph*)]
    \item The transition probabilities satisfy: $\lambda_i\geq\lambda_o$.
    \item The costs satisfy: $0\leq C_o\leq C_i\leq C_c$.
\end{enumerate}
\end{assumption}
The first assumption 1.a) states that the patient's health improves faster under intensive monitoring, rather than under ordinary. Regarding assumption 1.b), it is naturally expected  that  $C_o\le C_i$, as the treatment burden is higher under intensive monitoring. Further, given the severity of entering the critical state $h=0$, it is naturally expected that $C_i\le C_c$, and practically $C_c$ is expected to be {\em much larger} than $C_i$. 

\subsection{Optimal Monitoring Control}
\begin{figure*}[h!]
\centering
\begin{subfigure}[t]{0.35\textwidth}
  \centering
  \includegraphics[width=0.8\linewidth]{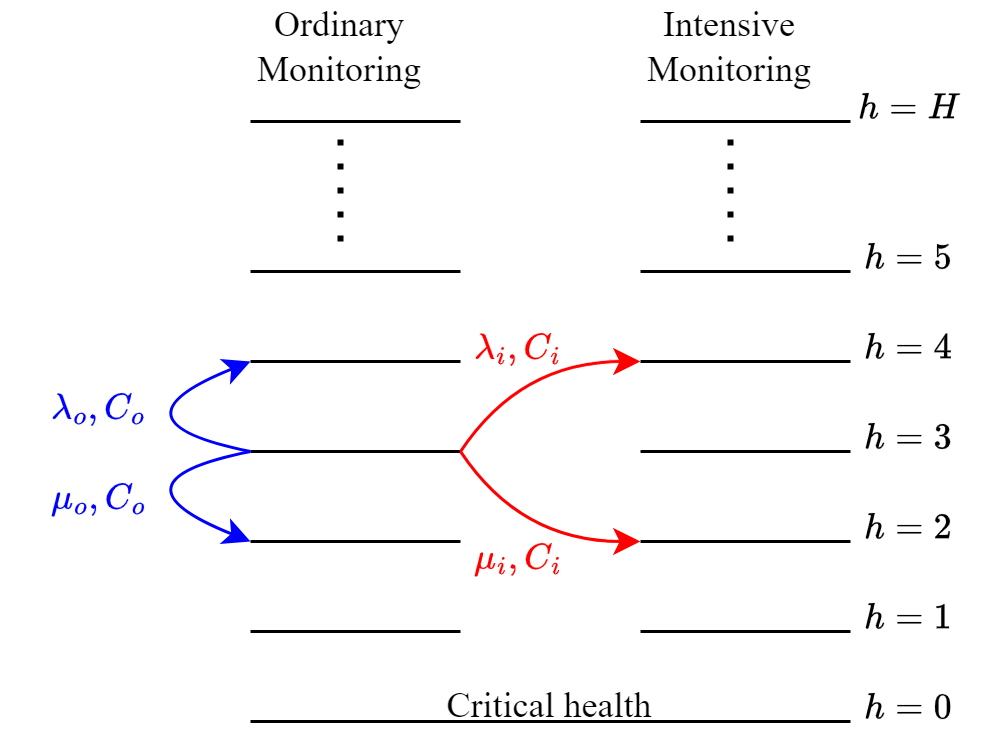}
  \caption{The Remote Patient Monitoring Model}
  \label{fig:rpm_model}
\end{subfigure}%
\hspace{5em}
\begin{subfigure}[t]{.35\textwidth}
  \centering
  \includegraphics[width=.75\linewidth]{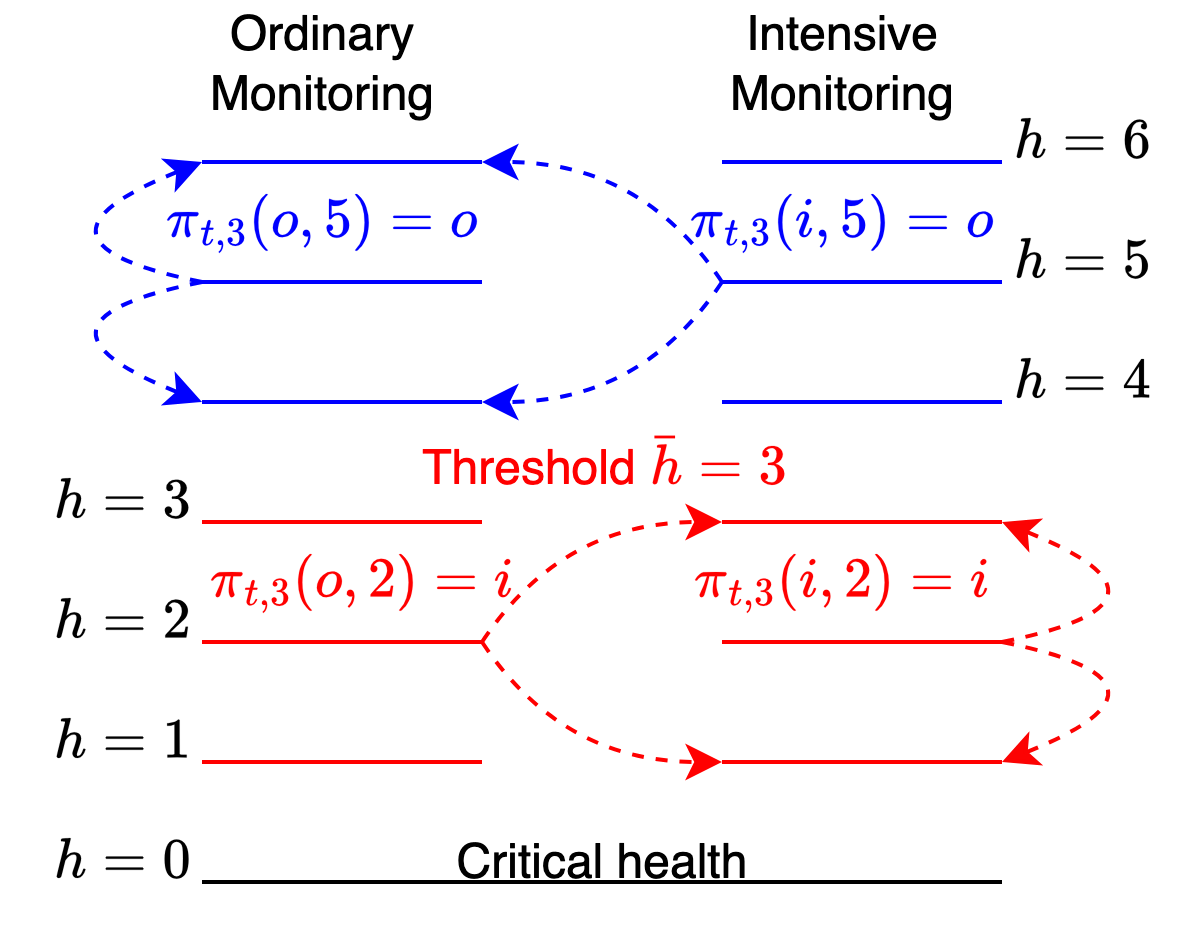}
  \caption{Threshold Policy $\pithresh$ with $\bar{h}=3$}
  \label{fig:threshold}
\end{subfigure}
\caption{(a) The RPM service model. The blue and red arrows represent the possible transitions from state $(o,3)$ for actions $o$ and $i$, respectively. The arrows are labeled with probability of transition and cost incurred. (b) An example where the optimal policy is the threshold policy, i.e., ordinary monitoring is used for $h>3$ and intensive for $h\leq3$.}
\label{fig:setup}
\end{figure*} 
We study this problem under the {\em discounted cost} setting of the dynamic programming (DP) methodology \cite{Bertsekas}.  Costs incurred $t$ time periods into the future are discounted by a factor of $\gamma^t$ with $0<\gamma<1$. Starting from state $s_0=s\in\Scal$, the total expected (discounted) cost to be incurred is
\begin{equation*}
\EE\Big[\sum_{t=0}^{T-1} \gamma^tc(s_t,a_t)+\gamma^TC_c \Big | \ s_o=s\Big]    
\end{equation*}
At each time $t$, control action $a_t$ is taken while in state $s_t=(m_t,h_t)$, at cost $c(s_t,a_t)$. The service ceases operation when the patient enters the critical state $h=0$ at time $T$, where $T$ is random and determined by evolution of the patient's health. At time $T$, the critical cost $C_c$ is incurred, discounted to $\gamma^TC_c$. Thus, discounting by $\gamma$ implicitly reflects the patient's desire to stay away from the critical health states. 

A stationary monitoring policy $\pi$ is a rule mapping each state $s\in\Scal$ to a control action $a\in\Acal=\{o,i\}$ to be taken at that state. The value function $V_\pi(s)$ of a policy $\pi$ is the total expected discounted cost the patient will incur until reaching the critical state, when it starts from state $s$ at time $t=0$:
$$
V_\pi(s)=\EE\left[\sum_{t=0}^{T-1}
\gamma^t c\Big(s_t,\pi(s_t)\Big)
+\gamma^TC_c\ \Big |\ s_0=s\right].
$$

The goal is to find an optimal control $\pi^*$ which minimizes $V_\pi$ over all controls $\pi$, i.e., $V_{\pi^*}(s) \le V_\pi(s)$ for every $s\in\Scal$ over all controls $\pi$. We define $V^*(s)\coloneqq V_{\pi^*}(s)$ which satisfies the following dynamic programming equation \cite{Bertsekas}.
\begin{equation*}
\Vstar(s)=\min_{a\in\{o,i\}} \Bigg\{c(s,a)+\gamma
\sum_{s'\in\Scal}\PP\Big(s'|s,a\Big)\Vstar(s')\Bigg\}, 
\end{equation*}
and can be solved numerically to yield the optimal policy $\pi^*(s)=\pi^*(m,h)$, that is, what optimal decision to take when the patient is in health state $h$ under monitoring $m$. For our model, the dynamic programming equation is as follows:
\begin{enumerate}
[align=left, 
leftmargin=0pt, 
labelindent=\parindent, listparindent=\parindent, 
labelwidth=0pt, 
itemindent=!,
label=(\roman*)]
\item 
At critical health $h=0, 
\Vstar(i,0)=\Vstar(o,0)=C_c,
$
\item 
For health states $1\leq h\leq H$, 
\begin{multline*}
\hskip -0.3 cm 
\Vstar(i,h)=\Vstar(o,h) \\
\hskip -0.3cm
=\min\Bigg\{ C_i + \gamma \Big[ \lambda_i\Vstar(i,h^+) + \mu_i \Vstar(i,h^-) \Big], \\
\hskip 1.35cm
C_o + \gamma \Big[ \lambda_o \Vstar(o,h^+) + \mu_o\Vstar(o,h^-)\Big] \Bigg\},
\end{multline*}
where $h^+=\min\{h+1,H\}$ and $h^-=h-1$.
\end{enumerate}

We next define a \textbf{threshold policy} $\pithresh$, characterized by the threshold $\hthresh$. Under these policies, there exists a threshold $\hthresh$ such that the patient stays in ordinary monitoring only when their health is better than $\hthresh$ (see Figure \ref{fig:threshold} for an example). So $\pithresh(i,h)=\pithresh(o,h)=i$ for $1\leq h\leq \hthresh$ and $\pithresh(i,h)=\pithresh(o,h)=o$ for $h>\hthresh$. The following result makes an important observation about the optimal control.
\begin{theorem}
    Under the assumption that $H$ and $1/\gamma$ are sufficiently large, the optimal policy $\pistar$ is a threshold-based policy $\pithresh$ for some $0\leq\hthresh\leq H$.
\end{theorem}
\begin{proof}
    See Appendix I.
\end{proof}
This result implies that we can reduce our problem to identifying optimal threshold $\hthresh$ by restricting our attention to policies of the form $\pithresh$. This result has been further corroborated by numerical experiments in \cite{chandak1, chandak2}. 

\section{Learning the Unknown RPM Model}
In the scenario where the parameters of the RPM service model are known to the clinicians, they can just obtain the optimal policy by solving the dynamic programming equations above. But in reality, parameters such as the costs ($C_o$ and $C_i$) and the probabilities ($\lambda_o$ and $\lambda_i$) are not accurately known to the clinicians. A rough estimate of the parameters might be available based on similar implementations of RPM service, but it is necessary to learn the true parameters to achieve the optimal benefits from the RPM service. 

The manner in which the samples are obtained for learning parameters depends on the exact service model and the definition of the parameters. If the cost parameters reflect the treatment burden on the patient, then the samples for estimating the invasiveness cost can be obtained by surveying patients about their experience with the monitoring service. In contrast, estimating the probability parameters is more straightforward, as it only requires observing changes in patient health state under different monitoring schemes. We assume that the critical cost is known to the service but it can be estimated in the same way as other costs.

It is important to note that these parameters are learned by observing the behavior of actual patients under the RPM service. We model this as an online reinforcement learning (RL) problem \cite{Sutton}, where the control policy is continually updated based on the samples observed thus far. The monitoring control must be adjusted to enable efficient learning of the parameters while preventing patients from incurring excessive costs. We propose a novel exploration scheme in Algorithm \ref{algo-main}, tailored to the RPM setting, that explicitly prioritizes patient safety during exploration.

The implementation of the algorithm is organized into epochs, each corresponding to a fixed time period (e.g., 10 days or one month). At the start of each epoch, the policy is computed based on the current parameter estimates and then applied to all patients throughout that epoch. During an epoch, both the patient population and its size may change, as new patients enroll in the remote monitoring program and existing patients leave upon reaching the critical health state.

There are three parts to the algorithm. In the beginning of each epoch, the system parameters are updated based on the samples collected till now. Next, a policy is computed for the RPM service based on the parameter estimates. Finally, this control is used in remote monitoring of patients, and samples are obtained for cost and transition probabilities. 

\begin{algorithm}[H]
\caption{Online Model‑Based Learning for RPM}
\begin{algorithmic}[1]\label{algo-main}
  \STATE \textbf{Input:} Priors $\lambda_{o,p}, \lambda_{i,p}, C_{o, p}, C_{i, p}$, prior strength $n_0$, healths $\hhigh$ and $\hlow$, probability sequence $\rho_k$ and critical cost $C_c$
  \STATE \textbf{Initialization:} Counts $N_a\gets n_0$, $N^+_a\gets \lambda_{a,p}n_0$, cost sums $S_a\gets C_{a,p}n_0$ for $a\in\{i,o\}$
  \STATE \textbf{for} epoch $k=1,\dots,$ \textbf{do}
    \STATE \quad \textbf{(1) Parameter Estimation:}
    \STATE \quad for $a\in\{i,o\}$, compute estimates $$\widehat\lambda_a\gets N^+_a/N_a\text{\;\;and\;\;}\widehat{C}_a\gets S_a/N_a$$
          \vspace{-16pt}
      \STATE \quad \textbf{(2) Safe and Exploratory Monitoring Control}
      \STATE \quad Compute optimal threshold based on current parameter estimates, denote by $\hat h_k$
      \STATE \quad Deployed threshold $\tilde h_k$:
      \begin{itemize}
          \item if $\hlow\leq \hat h_k\leq \hhigh$, $\tilde h_k=\hat h_k$,
          \item if $\hat h_k<\hlow$, $\tilde h_k = 
\begin{cases}
\hat h_k, & \text{w.p.}\;\; 1-\rho_k, \\
\hlow, & \text{w.p.}\;\; \rho_k.
\end{cases}$
          \item if $\hat h_k>\hhigh$, $\tilde h_k = 
\begin{cases}
\hat h_k, & \text{w.p.}\;\; 1-\rho_k, \\
\hhigh, & \text{w.p.}\;\; \rho_k.
\end{cases}$
      \end{itemize}
    \STATE \quad \textbf{(3) RPM Service and Data Collection:}
    \STATE \quad \textbf{for} each patient in system \textbf{do}
    \STATE \quad \quad \textbf{for} timesteps $t=0,\ldots,T-1$ within epoch $k$ \textbf{do}
      \STATE \quad \quad \quad Observe health state $h_t$
      \STATE \quad \quad \quad Implement monitoring control $a=\pi_{t,\tilde h_k}(h_t)$
      \STATE \quad \quad \quad Observe next state $h_{t+1}$ and cost sample $c_t$
      \STATE \quad \quad \quad Update: 
      \vspace{-5pt}
      \begin{align*}
          N^+_a&\gets N^+_a+1\text{\;if\;}h_{t+1}=h_t+1 \text{ or } H\\
          N_a&\gets N_a+1, S_a\gets S_a+c_t.
      \end{align*}
      \vspace{-20pt}
      \STATE \quad \quad \textbf{end for}
      \STATE \quad \textbf{end for}
  \STATE \textbf{end for}
  \STATE \textbf{Output:} Final threshold $\hat h$ and estimates $\hat\lambda_a,\hat C_a$.
\end{algorithmic}
\end{algorithm}

Before discussing the intuition behind Algorithm \ref{algo-main}, we discuss some notation used therein. For $a\in\{i,o\}$, $N_a$ denotes the number of times monitoring $a$ has been applied, $N_a^+$ denotes the number of times the patient's health improved under monitoring $a$, and $S_a$ denotes the total cost patients had to incur under monitoring $a$. These are used to estimate the transition probabilities $\widehat\lambda_a=N_a^+/N_a$ and costs $\widehat C_a=S_a/N_a$. These estimates are initialized using the priors $\lambda_{a,p}$ and $C_{a,p}$ and the strength in these beliefs $n_0$. 

In each epoch, we compute the optimal threshold corresponding to the current parameter estimates and denote it by $\hat h_k$. The threshold actually deployed in epoch $k$ is denoted by $\tilde h_k$. This deployed threshold may differ from $\hat h_k$ according to the following simple scheme, which is designed to ensure sufficient exploration. If the computed optimal threshold $\hat h_k$ is too low, the system may collect too few samples from patients receiving intensive monitoring. Therefore, when $\hat h_k<\hlow$, we set $\tilde h_k$ to $\hlow$ with probability $\rho_k$. Similarly, when $\hat h_k>\hhigh$, we set $\tilde h_k$ to $\hhigh$ with probability $\rho_k$. The probability sequence $\rho_k$ is chosen to be slowly decreasing, so that the algorithm explores sufficiently in early epochs and gradually relies more on the estimated optimal threshold as the parameter estimates become more accurate.

This exploration scheme has three main advantages. First, unlike generic exploration methods in reinforcement learning, which typically perturb arbitrary state-action decisions, our scheme exploits the threshold structure of the RPM problem. It acts directly on clinically meaningful quantities, making the algorithm interpretable and easy to implement within clinical workflows. In particular, the parameters $\hlow$ and $\hhigh$ correspond to fragile and sufficiently good health states, while $\rho_k$ can be chosen as any slowly decreasing sequence. Second, because exploration is performed by adjusting the  threshold rather than arbitrary actions, the scheme naturally supports patient-safety guarantees. Third, under mild assumptions, the resulting policy converges almost surely to the optimal policy. The following theorem establishes convergence under standard technical assumptions. The assumptions and proof are provided in Appendix~II.

\begin{theorem}\label{thm:convergence}
The policy (equivalently, the threshold) deployed by Algorithm~\ref{algo-main} converges almost surely to the optimal policy (respectively, optimal threshold).
\end{theorem}

\begin{figure*}[h!]
\centering
\begin{subfigure}[t]{0.4\textwidth}
  \centering
  \includegraphics[width=0.95\linewidth]{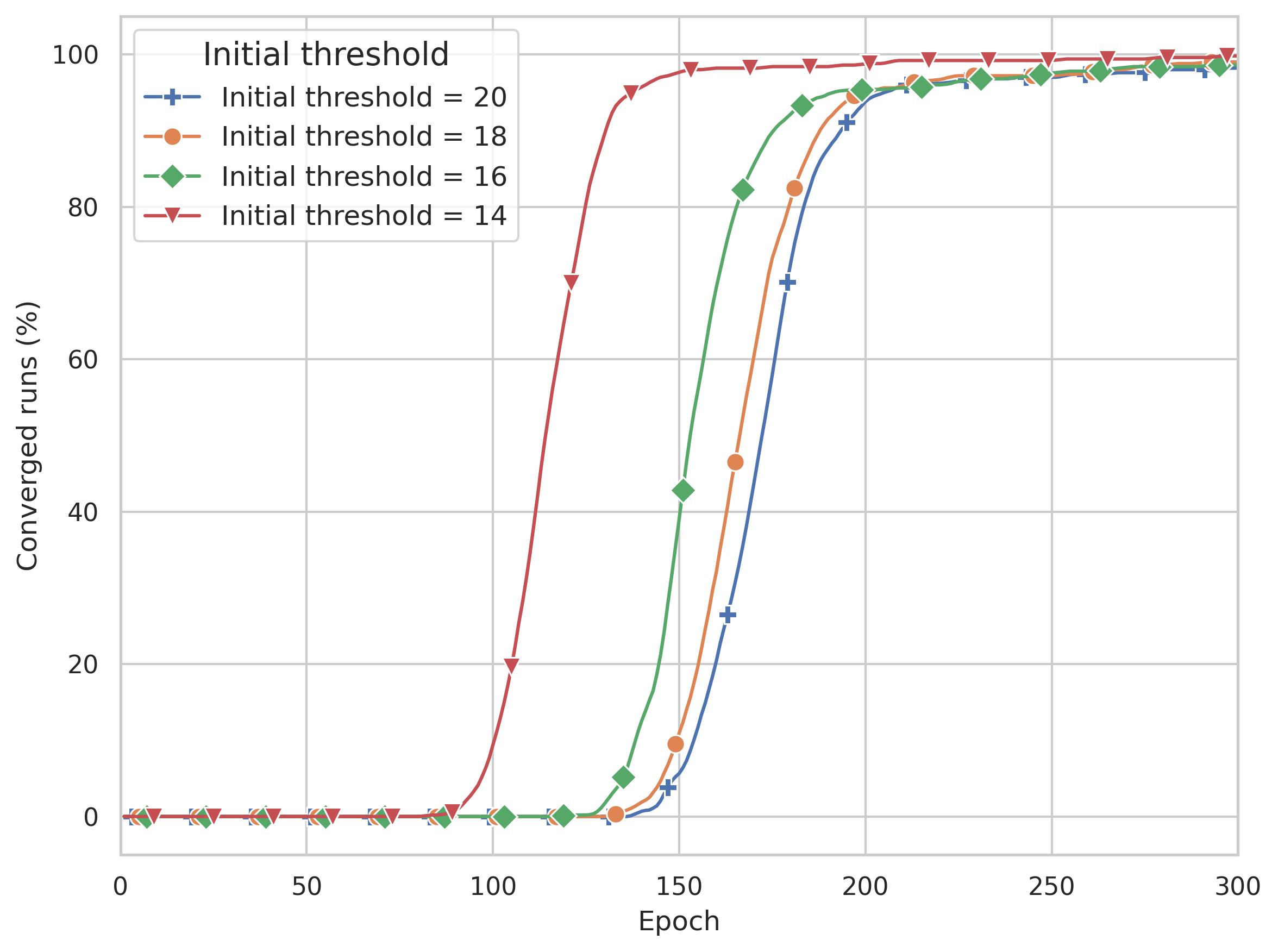}
  \caption{Effect of initial threshold}
  \label{fig:convergence_prob_vs_initial_threshold}
\end{subfigure}%
\begin{subfigure}[t]{0.4\textwidth}
  \centering
  \includegraphics[width=0.95\linewidth]{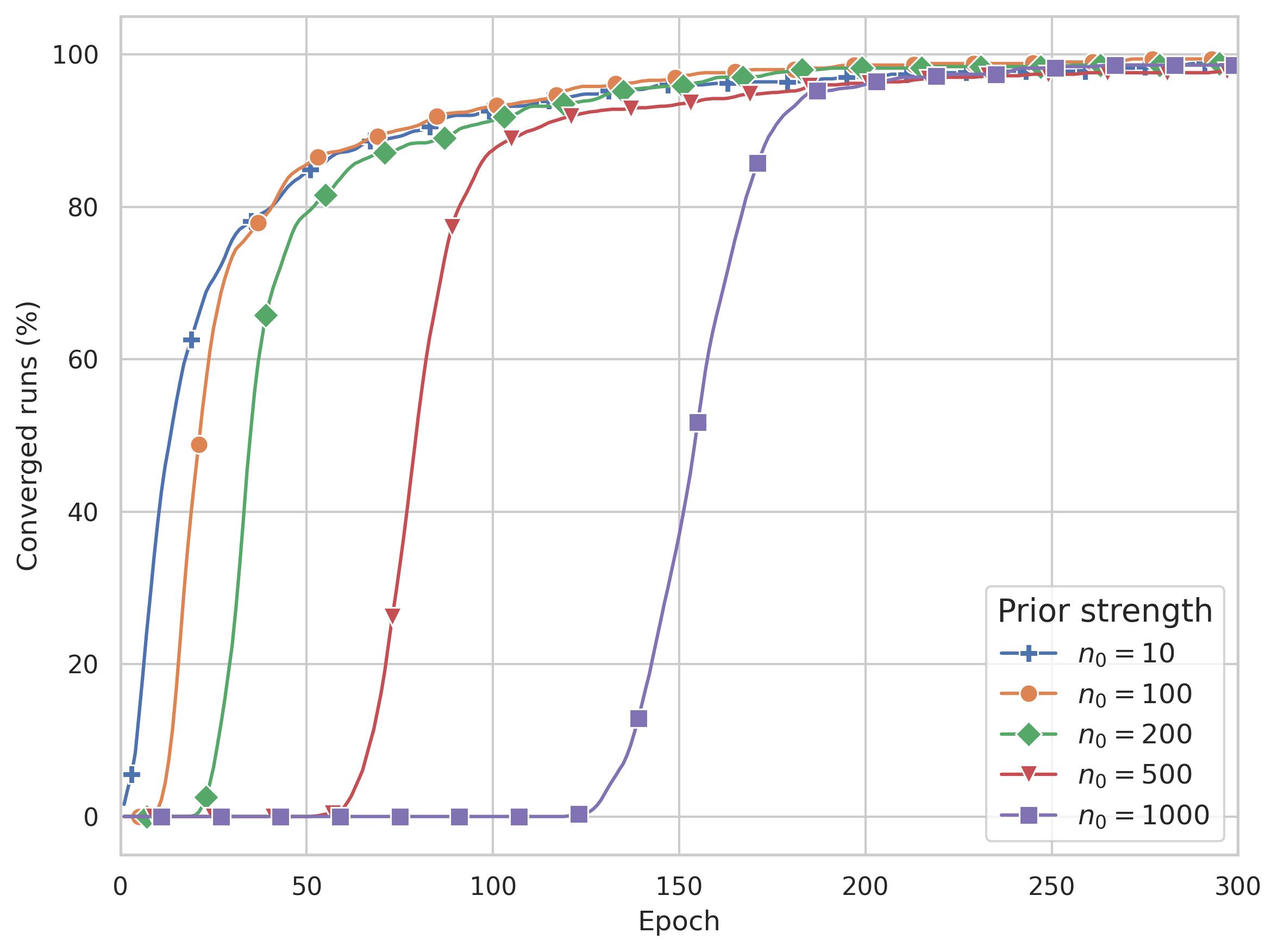}
  \caption{Effect of prior strength}
  \label{fig:convergence_prob_vs_prior_strength}
\end{subfigure}%
\caption{Percentage of converged runs after each epoch for Algorithm \ref{algo-main}}
\label{fig:convergence}
\end{figure*}

\subsection{Safety Considerations}\label{subsec:safe}
As the parameters are learned while the RPM service is deployed on real patients, it is essential that the policy avoid decisions that increase the risk of patients reaching critical health states. One extreme way to ensure safety is to place all patients under intensive monitoring at all health states. However, this slows parameter learning and substantially increases monitoring costs. With appropriate configuration, our algorithm balances safety and efficiency. In particular, the algorithm should be initialized by carefully choosing the priors $\lambda_{a,p}, C_{a,p}$ and the prior strength $n_0$.

The priors should be chosen to induce a conservative initial policy, since they determine the controls applied before sufficient data are collected. When rough parameter estimates are available, they can be used as priors. Otherwise, we recommend setting $\lambda_{i,p}\gg\lambda_{o,p}$ and $C_{i,p}\approx C_{o,p}$, which makes intensive monitoring more favorable, delaying the occurrence of critical states at the expense of higher monitoring cost. The prior strength $n_0$ controls how quickly the policy adapts: larger values slow parameter updates, prevent abrupt early policy shifts, and keep the policy close to its conservative initial form. The following theorem shows that if the priors are chosen as described above and the prior strength $n_0$ is sufficiently large, then, with high probability, the deployed threshold is never lower than the optimal threshold throughout the learning process. The assumptions and proof are provided in Appendix~III.
\begin{theorem}\label{thm:safety}
For any $\delta\in(0,1)$, if the prior strength satisfies $n_0 \geq \Omega(\log(1/\delta))$, then, with probability at least $1-\delta$, the deployed threshold satisfies $\tilde h_k\geq \hstar$ for all epochs $k$.
\end{theorem}

The above result states that, with probability at least $1-\delta$, the deployed policy is at least as safe as the optimal policy at every epoch, where safety is measured by the time taken to reach the critical state. Beyond the guarantees of the theorem, the exploration mechanism itself provides an additional safeguard: when the estimated optimal threshold falls below $\hlow$, the deployed threshold is raised to $\hlow$ with probability $\rho_k$, reducing the likelihood that the algorithm applies overly aggressive low-threshold policies during learning.

\section{Performance} \label{sec:performance}

In this section, we evaluate the performance and safety aspects of our algorithm through simulations. We examine its convergence to the optimal policy, and analyze how prior strength and initial threshold affect performance and safety.
 
To evaluate our algorithm, we simulate a Markov chain for each patient. Unless stated otherwise, each epoch of the algorithm consists of 75 timesteps distributed across patients. The plots of average time to reach the critical health state are obtained by averaging results over 200 independent runs. To plot the percentage of converged runs across epochs, we conduct 10 experiments, each comprising 50 runs. For each experiment, we compute the percentage of runs that have converged after every epoch and then average these percentages across all experiments. In all the plots the initial thresholds are chosen to be larger than the optimal threshold. A run is said to have converged if all estimated parameters are within 5\% of their true values. Complete experimental details are provided in Appendix IV.

\begin{figure*}[h!]
\centering
\begin{subfigure}[t]{0.4\textwidth}
  \centering
  \includegraphics[width=0.95\linewidth]{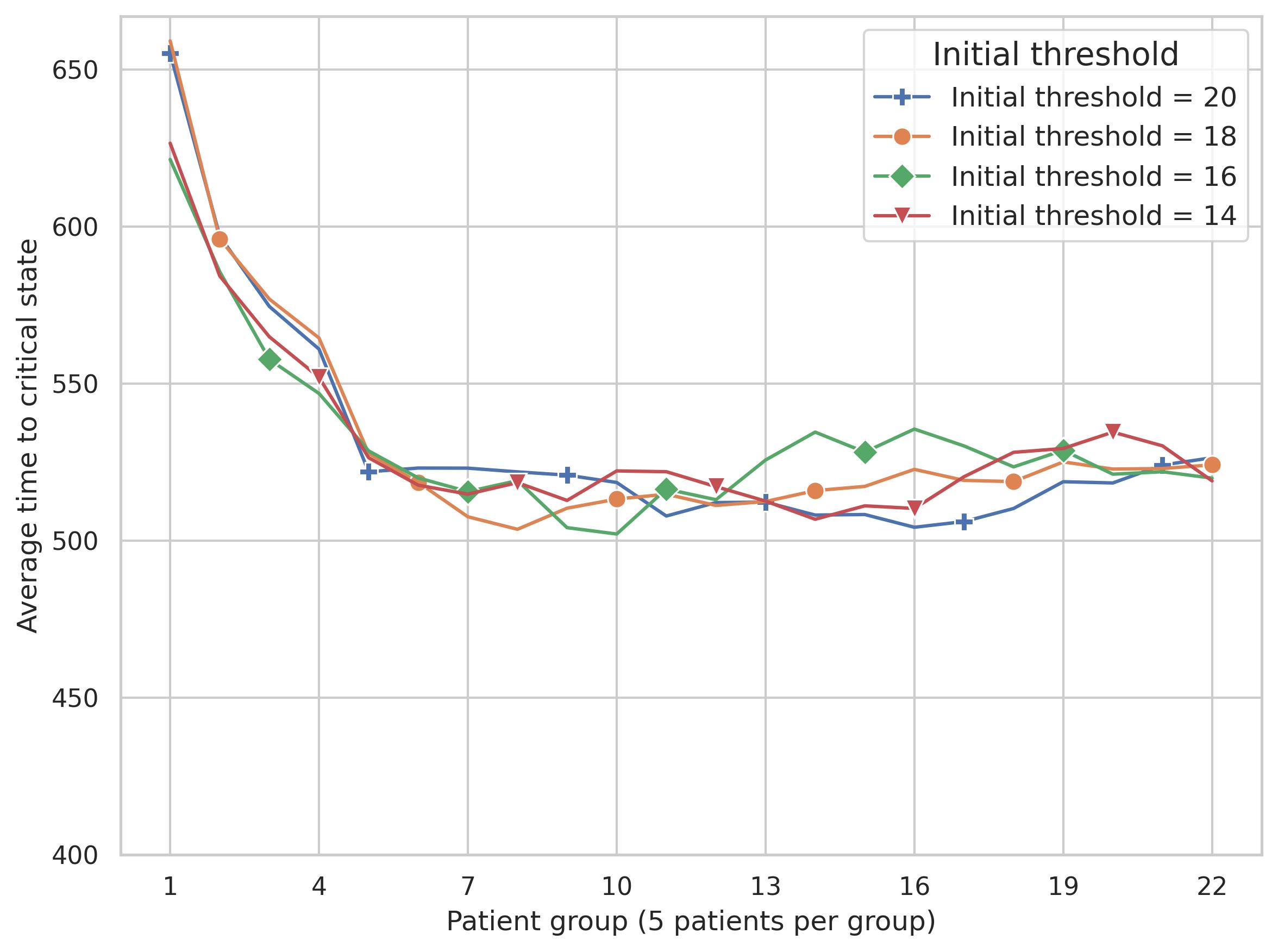}
  \caption{\centering Effect of initial threshold}
  \label{fig:avg_T_death_vs_initial_threshold}
\end{subfigure}%
\begin{subfigure}[t]{0.4\textwidth}
  \centering
  \includegraphics[width=0.95\linewidth]{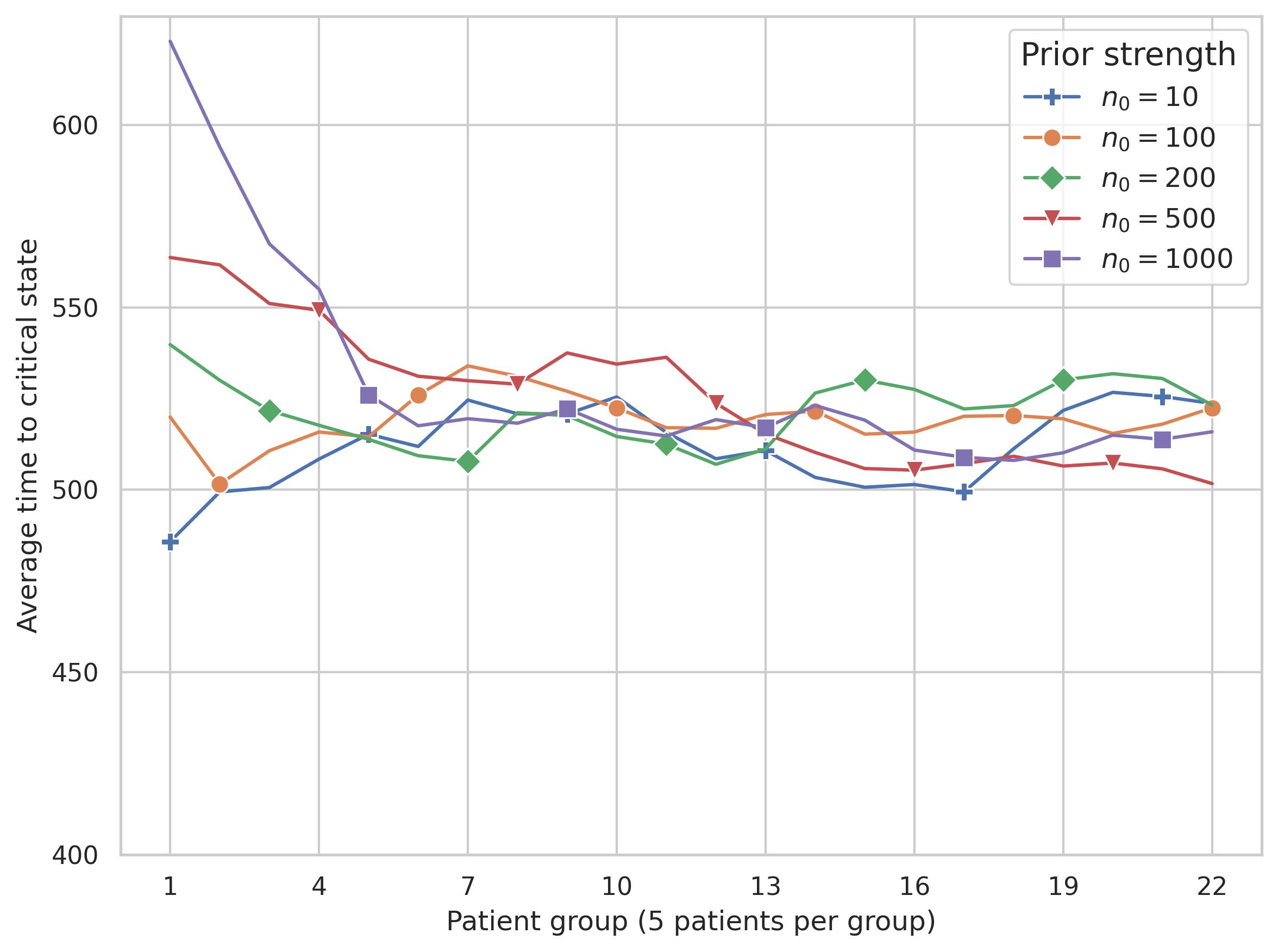}
  \caption{\centering Effect of prior strength}
  \label{fig:avg_T_death_vs_prior_strength}
\end{subfigure}
\caption{Average time to reach critical health for Algorithm~\ref{algo-main}. Patients are grouped only for reporting: group 1 contains the first five patients to enter the system, group 2 contains the next five, and so on. These groups do not represent synchronous cohorts, and patients within a group may enter and leave the system at different times.}
\label{fig:time_to_critical}
\end{figure*}

We observe that the algorithm converges to the optimal or a near-optimal policy within a few epochs in almost all runs. This behavior remains consistent across a wide range of parameter settings. In Figure~\ref{fig:convergence}, we plot the percentage of runs that have converged after each epoch. As expected, parameter estimates take longer to converge to the true values when the initial threshold is farther from the true threshold and when the prior strength is larger. However, this slower convergence of the parameter estimates has little practical impact, since the threshold policy is relatively robust to inaccuracies in the learned RPM parameters, allowing the algorithm to attain a near-optimal policy well before the underlying estimates have fully converged.

Figure~\ref{fig:time_to_critical} shows the average time to reach the critical health state under different initial thresholds and prior strengths. Across initial thresholds, the algorithm converges to the optimal policy $(\hthresh=12)$ after approximately the same number of epochs, indicating robustness to initialization. The results also show that, for sufficiently large prior strengths, the average time to reach the critical health state remains consistently high throughout learning, indicating that the algorithm maintains patient safety while adapting the monitoring policy.

\begin{table}[t]
\centering
\caption{Table for safety of the deployed threshold for different prior strengths. Each entry reports the number of runs, out of 100 runs, in which the deployed threshold fell below the corresponding cutoff in at least one epoch.}
\label{tab:safety}
\resizebox{\columnwidth}{!}{
\begin{tabular}{c c c c c c c}
\toprule 
\(n_0\) & \(h^\star\) & \(h_{\mathrm{init}}\) & \(\tilde h<h^\star\) & \(\tilde h<h^\star-1\) & \(\tilde h<h^\star-2\) & \(\tilde h<\hlow\) \\ 
\midrule
10   & 12 & 16 & 82 & 38 & 8 & 0 \\
100  & 12 & 16 & 50 & 1  & 0 & 0 \\
200  & 12 & 16 & 25 & 0  & 0 & 0 \\
500  & 12 & 16 & 6  & 0  & 0 & 0 \\
1000 & 12 & 16 & 0  & 0  & 0 & 0 \\
\bottomrule
\end{tabular}
}
\end{table}

When the algorithm is initialized with a sufficiently conservative threshold and large prior strength, the deployed threshold remains above the optimal threshold in almost every epoch, as shown in Table~\ref{tab:safety}. This behavior is consistent with the safety guarantee in Theorem~\ref{thm:safety}. We also report more severe threshold violations, where the deployed threshold falls below $h^\star-1$, $h^\star-2$, or $\hlow$. The number of such violations drops sharply as the severity of the violation increases. For example, across 100 independent simulation runs, even for a relatively small prior strength of $n_0=100$, the deployed threshold falls below $h^\star$ in 50 runs, but falls below $h^\star-1$ in only one run and never falls below $h^\star-2$ or $\hlow$. Increasing the prior strength further reduces the frequency of threshold violations, making the deployed policy increasingly conservative during the learning phase.

\section{Conclusions and Extensions}\label{sec:conclusions}
We develop an online model-based reinforcement learning approach for optimizing RPM policies when system parameters are initially unknown. By combining parameter estimation with a simple exploration scheme, our algorithm efficiently balances learning and patient safety. Theoretical guarantees and simulations demonstrate that the proposed approach effectively learns the transition probabilities and costs, and converges to the optimal threshold policy. It maintains low overall costs and safe monitoring decisions, representing a promising step toward more flexible and adaptive RPM tools.

An important next step is to implement the proposed algorithm on real-world patient data to assess its practical performance. While our simulations demonstrate its effectiveness, testing on clinical data will reveal its robustness to real-world variability and noise. Such validation will be essential for integrating the approach into actual RPM workflows and guiding its deployment in clinical settings.

This paper focuses on the special case where RPM parameters are independent of health states. Although the proposed algorithm can be extended to health-state-dependent parameters, it is practical only for small state spaces, since the number of parameters and required samples grow quickly with the number and dimension of health states. Addressing this curse of dimensionality is an important future direction, potentially through approximate dynamic programming, deep reinforcement learning, or structural assumptions such as smoothness and parameter sharing across neighboring health states \cite{Sutton,Bertsekas,approx-dp,deep-rl}.

\bibliographystyle{IEEEtran}
\bibliography{refs}

\appendices

\section{Proof for Theorem 1}

\begin{proof}
Theorem 1 from \cite{chandak1} states that when $H\rightarrow\infty$, the optimal policy is $\pi_o$ when \begin{equation}\label{cond-pi-o}
    \gamma(\lambda_i-\lambda_o)\leq C_i/C_c.
\end{equation} Here $\pi_o$ is the policy where the patient always stays in ordinary monitoring, i.e., threshold-based policy $\pithresh$ with $\hthresh=0$. Theorem 2 from \cite{chandak1} states that the optimal policy is $\pithresh$ with some $\hthresh$ when the following are satisfied - 
\begin{equation}\label{cond-pithresh-1}
    \gamma(\lambda_i-\lambda_o)> C_i/C_c,
\end{equation}
and
\begin{equation}\label{cond-pithresh-2}
    \frac{\gamma\mu_o(1+\gamma\mu_o)}{1-\gamma^2\lambda_o\mu_o}\leq 1.
\end{equation}
Note that \eqref{cond-pithresh-1} is the complement of \eqref{cond-pi-o} and hence one of them is always true. Hence we only need to verify that equation \eqref{cond-pithresh-2} is satisfied. Condition \eqref{cond-pithresh-2} is satisfied when $$\gamma\leq \frac{-1+\sqrt{1+4/\mu_o}}{2},$$ which is satisfied for sufficiently large $1/\gamma$. This completes the proof for Theorem 1.
\end{proof}

\section{Convergence Guarantees}
We use the same model, notation, threshold policy convention, and epoch wise learning algorithm as defined in the main paper. In particular, $h^\star$ denotes the optimal threshold under the true parameters, $\hat h_k$ denotes the optimal threshold computed corresponding to the current parameter estimates at epoch $k$. We now state the additional assumptions needed for the safety and almost sure convergence arguments.

\begin{assumption}[]\label{ass:finite-state-epoch}
The health-state space is finite and is given by $\Hcal=\{0,1,\ldots,H\}$, where $0$ is the critical state and $H$ is the maximum health state. Each epoch consists of $L$ patient-state transitions, where $L>H$. Moreover, the one-step upward and downward movement probabilities under both actions are positive, i.e., $p_{\min}=\min\{\lambda_i,\lambda_o,\mu_i,\mu_o\}>0.$
\end{assumption}

\begin{assumption}[]\label{ass:rho-diverge}
At epoch $k$, if the computed threshold $\hat{h}_k$ lies outside the interval $[\hlow, \hhigh]$, the algorithm implements the threshold $\hlow$ or $\hhigh$  with probability $\rho_k$. These probabilities satisfy
$\sum_{k=1}^{\infty}\rho_k=\infty.$
\end{assumption}

\begin{assumption}[]\label{ass:local-unique}
Let $\theta=(\lambda_i,\lambda_o,C_i,C_o)$ denote the true parameter vector. There exists $\varepsilon_{\mathrm{pol}}>0$ such that for every parameter vector $\theta'$ satisfying $\|\theta'-\theta\|_\infty<\varepsilon_{\mathrm{pol}}$,
the unique optimal threshold under $\theta'$ is $\hstar$.
\end{assumption}

\begin{assumption}[]\label{ass:obs-prior}
The cost noise under each action $a\in\{i,o\}$ is sub-Gaussian. That is, if $G_{a,t}$ denotes the cost observed on the $t$-th real sample of action $a$, then $G_{a,t}-C_a$ is sub-Gaussian. Let $\sigma^2<\infty$ denote a common conditional sub-Gaussian variance proxy for the cost noises and the transition observation noises.
\end{assumption}
We can now state the main result of this appendix.
\begin{theorem}[]\label{thm:as-policy}
Suppose $\hlow\leq\hstar\leq \hhigh$. Under Assumptions~\ref{ass:finite-state-epoch}, \ref{ass:rho-diverge}, \ref{ass:local-unique}, and~\ref{ass:obs-prior}, there exists an almost surely finite random epoch $K_G$ such that, for all $k\geq K_G$, the threshold deployed by the algorithm at epoch $k$ is equal to the optimal threshold $\hstar$.
\end{theorem}

\subsection{Proof}
We first prove that each action is sampled infinitely often. Then we use this to prove convergence of the parameter estimates to the true parameters, eventual correctness of the computed threshold, and finally eventual correctness of the deployed threshold.
\begin{lemma}[]\label{lem:sample-count-as}
Let $X_{a,k}$ denote the number of real samples of action $a\in\{i,o\}$ collected during epoch $k$ and $p_{\min}=\min\{\lambda_i,\lambda_o,\mu_i,\mu_o\}>0.$
Then, for each action $a\in\{i,o\}$, there exists $c_a>0$ such that
\begin{align}\label{eq:cond-sampling-bound}
\EE[X_{a,k}\mid\Fcal_{k-1}]\geq c_a\rho_k.
\end{align}
Moreover, for each action $a\in\{i,o\}$, $N_{a,j}\to\infty\text{ a.s.}$
\end{lemma}

\begin{proof}
Fix an action $a\in\{i,o\}$ and an epoch $k$. With a probability at least $\rho_k$, the deployed threshold has at least one health state $s_a$ at which action $a$ is used. This is possible because $\hlow\geq1$ and $\hhigh\leq H-1$, so the deployed thresholds are such that both actions have nonempty state regions. Now fix any starting health state $s\in\{1,\ldots,H\}$. Since the state space has size $H$, the target state $s_a$ can be reached from $s$ in at most $H-1$ one-step moves by moving monotonically toward $s_a$. Each required one-step move has probability at least $p_{\min}$. Since $L> H$, this path fits inside one epoch. Hence, conditional on using such a threshold, the probability of visiting a state where action $a$ is used during the epoch is at least $p_{\min}^H$. Therefore,
$\PP(X_{a,k}\geq1\mid\Fcal_{k-1})\geq \rho_k p_{\min}^H.$
Since $X_{a,k}$ is a nonnegative integer-valued random variable,
\begin{align*}
\EE[X_{a,k}\mid\Fcal_{k-1}]
&\geq \PP(X_{a,k}\geq1\mid\Fcal_{k-1})\\
&\geq \rho_k p_{\min}^H.
\end{align*}
Thus first part of the lemma holds with $c_a=p_{\min}^H>0$. Since each epoch has $L$ patient-state transitions, $0\leq X_{a,k}\leq L.$ Since $X_{a,k}\leq L\mathbf{1}_{\{X_{a,k}\geq1\}}$, we have
\begin{align*}
\EE[X_{a,k}\mid\Fcal_{k-1}]\leq L\PP(\{X_{a,k}\geq1\}\mid\Fcal_{k-1}).
\end{align*}
Therefore, by \eqref{eq:cond-sampling-bound},
\begin{align*}
\PP(\{X_{a,k}\geq1\}\mid\Fcal_{k-1})\geq\frac{1}{L}\EE[X_{a,k}\mid\Fcal_{k-1}]\geq\frac{c_a}{L}\rho_k.
\end{align*}
Since $\sum_{k=1}^{\infty}\rho_k=\infty$, it follows that
\begin{align*}
\sum_{k=1}^{\infty}\PP(\{X_{a,k}\geq1\}\mid\Fcal_{k-1})=\infty\text{ a.s.}
\end{align*}
By the extended Borel-Cantelli lemma~\cite[Corollary 5.29]{breiman1992probability}, the events $\{X_{a,k}\geq1\}$ occur infinitely often almost surely. Hence $\sum_{k=1}^{j}X_{a,k}\to\infty\text{ a.s.}$
Since $N_{a,j}=n_0+\sum_{k=1}^{j}X_{a,k},$
we get $N_{a,j}\to\infty\text{ a.s.}$ This proves the lemma.
\end{proof}

\begin{lemma}[]\label{lem:param-consistency}
The parameter estimates converge to the true parameter vector almost surely, i.e., $\thetahat_k\to\theta \text{ a.s.}$ Moreover, there exists an almost surely finite random epoch $K_G$ such that $\hhat_k=\hstar \text{ for all }k\geq K_G.$
\end{lemma}

\begin{proof}
Fix an action $a\in\{i,o\}$. Let $M_{a,k}=N_{a,k}-n_0$ denote the number of real samples of action $a$ collected up to the end of epoch $k$. By Lemma~\ref{lem:sample-count-as}, $M_{a,k}\to\infty \text{ a.s.}$ Let $Y_{a,t}$ be the upward-movement indicator observed on the $t$-th real sample of action $a$. $\EE[Y_{a,t}\mid\Fcal_{a,t-1}]=\lambda_a,$ where $\Fcal_{a,t-1}$ denotes the information available before the $t$-th real sample of action $a$. Therefore $Y_{a,t}-\lambda_a$ is a bounded martingale difference sequence. By the martingale strong law of large numbers~\cite{atchade2009slln},
\begin{align*}
\frac{1}{n}\sum_{t=1}^n(Y_{a,t}-\lambda_a)\to0 \qquad\text{a.s.}
\end{align*}
Since $M_{a,k}\to\infty$ almost surely, the same convergence holds along the random subsequence $n=M_{a,k}$,
\begin{align*}
\frac{1}{M_{a,k}}\sum_{t=1}^{M_{a,k}}Y_{a,t} \to\lambda_a \qquad\text{a.s.}
\end{align*}
The prior-weighted transition estimate can be written as
\begin{align*}
\hat\lambda_{a,k} = \frac{n_0\lambda_{a,p}+\sum_{t=1}^{M_{a,k}}Y_{a,t}} {n_0+M_{a,k}}.
\end{align*}
Since $n_0<\infty$, $\lambda_{a,p}$ is finite, and $M_{a,k}\to\infty$, the prior contribution vanishes.
Therefore $\hat\lambda_{a,k}\to\lambda_a \text{ a.s.}$ The same argument applies to the cost estimate. Let $G_{a,t}$ be the cost observed on the $t$-th real sample of action $a$. By Assumption~\ref{ass:obs-prior}, $\EE[G_{a,t}\mid\Fcal_{a,t-1}]=C_a,$
and $G_{a,t}-C_a$ is sub-Gaussian. In particular, it has uniformly bounded second moment. Hence the martingale strong law gives
\begin{align*}
\frac{1}{n}\sum_{t=1}^n(G_{a,t}-C_a)\to0 \qquad\text{a.s.}
\end{align*}
As before, evaluating this convergence along $n=M_{a,k}$ and as the prior weight is finite, $\hat C_{a,k}\to C_a \text{ a.s.}$ This holds for both $a=i$ and $a=o$. Therefore $\thetahat_k \to \theta \text{ a.s.}$
As $\thetahat_k\to\theta \text{ a.s.}$, there exists a finite random epoch $K_G$ such that for every $k\geq K_G$,
$\|\thetahat_k-\theta\|_\infty<\varepsilon_{\mathrm{pol}}.$
By Assumption~\ref{ass:local-unique}, whenever this inequality holds, the unique optimal threshold under $\thetahat_k$ is $\hstar$. So, it follows that $\hhat_k=\hstar \text{ for all }k\geq K_G.$
This proves the lemma.
\end{proof}

\begin{proof}[Proof of Theorem~\ref{thm:as-policy}]
By Lemma~\ref{lem:param-consistency}, there exists an almost surely finite random epoch $K_G$ such that $\hhat_k=\hstar \text{ for all }k\geq K_G.$ Since $\hlow\leq\hstar\leq \hhigh$, the optimal threshold lies in the region where the algorithm does not change the threshold. Hence, for all $k\geq K_G$, the threshold deployed by the algorithm at epoch $k$ is equal to the computed threshold $\hhat_k$, and both are equal to the optimal threshold $\hstar$.
\end{proof}

\section{Safety Guarantees}
The safety guarantee requires two additional assumptions stated below, beyond those stated in Appendix II. 

\begin{assumption}[]\label{ass:mono}
The optimal integer threshold $\hstar$ is non-decreasing in $\lambda_i$ and $C_o$, and non-increasing in $\lambda_o$ and $C_i$.
\end{assumption}

\begin{assumption}[]\label{ass:safe-prior}
The prior parameter vector $\theta_p=(\lambda_{i,p},\lambda_{o,p},C_{i,p},C_{o,p},C_c)$ is biased by $\Delta>0$ such that
$\lambda_{i,p}\geq\lambda_i+\Delta,\lambda_{o,p}\leq\lambda_o-\Delta,C_{i,p}\leq C_i-\Delta$ and $C_{o,p}\geq C_o+\Delta.$
\end{assumption}
We can now state the main result of this appendix.
\begin{theorem}[]\label{thm:global-safety}
Suppose Assumptions~\ref{ass:local-unique}, \ref{ass:obs-prior}, \ref{ass:mono}, and~\ref{ass:safe-prior} hold, and suppose $\hlow\leq\hstar\leq\hhigh$. Fix any $0<\varepsilon_{\mathrm{safe}}<\varepsilon_{\mathrm{pol}}.$
If
\begin{align*}
n_0\geq \frac{\sigma^2\log(4/\delta)} {2\varepsilon_{\mathrm{safe}}(\Delta+\varepsilon_{\mathrm{safe}})},
\end{align*}
then, with probability at least $1-\delta$, the threshold deployed by the algorithm at epoch $k$ is at least $\hstar$ for all $k\geq1$.
\end{theorem}

\subsection{Proof}
We first show that the local uniqueness assumption, together with monotonicity, implies a one-sided safety region around the true parameter vector.
\begin{lemma}[]\label{lem:one-sided-safety}
Fix any $0<\varepsilon_{\mathrm{safe}}<\varepsilon_{\mathrm{pol}}.$
If a parameter vector $\theta'=(\lambda'_i,\lambda'_o,C'_i,C'_o)$
satisfies
\begin{align*}
\lambda'_i\geq\lambda_i-\varepsilon_{\mathrm{safe}},\qquad \lambda'_o\leq\lambda_o+\varepsilon_{\mathrm{safe}},
\end{align*}
\begin{align*}
C'_i\leq C_i+\varepsilon_{\mathrm{safe}},\qquad C'_o\geq C_o-\varepsilon_{\mathrm{safe}},
\end{align*}
then the computed threshold under $\theta'$ is safe, i.e.,
$h(\theta')\geq\hstar.$
\end{lemma}

\begin{proof}
Define the boundary parameter vector
\begin{align*}
\theta^-= (\lambda_i-\varepsilon_{\mathrm{safe}}, \lambda_o+\varepsilon_{\mathrm{safe}}, C_i+\varepsilon_{\mathrm{safe}}, C_o-\varepsilon_{\mathrm{safe}}).
\end{align*}
Since $\varepsilon_{\mathrm{safe}}<\varepsilon_{\mathrm{pol}}$,
$\|\theta^--\theta\|_\infty = \varepsilon_{\mathrm{safe}} < \varepsilon_{\mathrm{pol}}.$ By Assumption~\ref{ass:local-unique}, the unique optimal threshold under $\theta^-$ is $\hstar$. Hence $h(\theta^-)=\hstar.$ Now let $\theta'$ satisfy the one-sided inequalities in the statement of the lemma. Compared with $\theta^-$, the parameter vector $\theta'$ has larger $\lambda_i$, smaller $\lambda_o$, smaller $C_i$, and larger $C_o$. By Assumption~\ref{ass:mono}, these changes can only increase the optimal threshold. Therefore $h(\theta')\geq h(\theta^-)=\hstar.$ This proves the lemma.
\end{proof}

\begin{proof}[Proof of Theorem~\ref{thm:global-safety}]
We prove that, with probability at least $1-\delta$, the parameter estimates remain in the one-sided safe region of Lemma~\ref{lem:one-sided-safety} for all epochs. We first consider the estimate of $\lambda_i$. Let $Y_t$ denote the upward movement indicator on the $t$-th real sample collected under action $i$. We know, $\EE[Y_t\mid\Fcal_{t-1}]=\lambda_i.$ Define $ S_n=\sum_{t=1}^n(Y_t-\lambda_i).$
Then $(S_n)_{n\geq0}$ is a martingale. The prior-weighted estimate after $n$ real samples of action $i$ is
\begin{align*}
\hat\lambda_i(n) = \frac{n_0\lambda_{i,p}+\sum_{t=1}^nY_t}{n_0+n}.
\end{align*}
By Assumption~\ref{ass:safe-prior}, $\lambda_{i,p}\geq\lambda_i+\Delta.$
Suppose that for some $n$, $\hat\lambda_i(n)<\lambda_i-\varepsilon_{\mathrm{safe}}.$
Then
\begin{align*}
n_0\lambda_{i,p}+\sum_{t=1}^nY_t < (n_0+n)(\lambda_i-\varepsilon_{\mathrm{safe}}).
\end{align*}
Using $\lambda_{i,p}\geq\lambda_i+\Delta$, this implies
\begin{align*}
S_n = \sum_{t=1}^n(Y_t-\lambda_i) < -n_0(\Delta+\varepsilon_{\mathrm{safe}}) -n\varepsilon_{\mathrm{safe}}.
\end{align*}
Since $Y_t-\lambda_i$ is conditionally sub-Gaussian with proxy $\sigma^2$, for every $\eta>0$,
\begin{align*}
\EE[\exp(-\eta (Y_t-\lambda_i))\mid\Fcal_{t-1}] \leq \exp\left(\frac{\eta^2\sigma^2}{2}\right).
\end{align*}
Hence
\begin{align*}
Z_n = \exp\left(-\eta S_n-\frac{\eta^2\sigma^2n}{2}\right)
\end{align*}
is a nonnegative supermartingale with $\EE[Z_0]=1$. Choose $\eta=\frac{2\varepsilon_{\mathrm{safe}}}{\sigma^2}.$
On the event
\begin{align*}
S_n<-n_0(\Delta+\varepsilon_{\mathrm{safe}})-n\varepsilon_{\mathrm{safe}},
\end{align*}
we have
\begin{align*}
-\eta S_n-\frac{\eta^2\sigma^2n}{2} > \eta n_0(\Delta+\varepsilon_{\mathrm{safe}}) + \left(\eta\varepsilon_{\mathrm{safe}}-\frac{\eta^2\sigma^2}{2}\right)n.
\end{align*}
By the chosen value of $\eta$, $\eta\varepsilon_{\mathrm{safe}}-\frac{\eta^2\sigma^2}{2}=0.$
Therefore, on the event that the estimate crosses the unsafe boundary, we have
\begin{align*}
Z_n> \exp\left(\eta n_0(\Delta+\varepsilon_{\mathrm{safe}})\right).
\end{align*}
By Ville's inequality for nonnegative supermartingales~\cite{howard2020timeuniform},
\begin{align*}
\PP\left( \exists n\geq1: \hat\lambda_i(n)<\lambda_i-\varepsilon_{\mathrm{safe}} \right) \leq \exp\left( -\eta n_0(\Delta+\varepsilon_{\mathrm{safe}}) \right).
\end{align*}
The same argument applies to the other three estimated coordinates, with the unsafe direction changed according to Assumption~\ref{ass:safe-prior}. By the assumed lower bound on $n_0$, each of the four probabilities is at most $\delta/4$. Therefore, by a union bound, with probability at least $1-\delta$, none of the four bad events occurs. On this good event, none of the four unsafe crossings occurs at any number of real observations and hence for all epochs. Thus $\thetahat_k$ satisfies the one-sided safe-region inequalities in Lemma~\ref{lem:one-sided-safety}. Hence $\hhat_k\geq\hstar\text{ for all }k\geq1.$ Since $\hlow\leq\hstar\leq\hhigh$, the algorithm never changes a computed threshold that is at least $\hstar$ to a threshold below $\hstar$. Therefore, the threshold deployed by the algorithm at epoch $k$ is at least $\hstar$ for all $k\geq1$.
\end{proof}

\section{Experiment Details}
In this section, we present complete experimental details for the plots presented in the main paper. We set $\hlow=4$, $\hhigh=17$ across all experiments in this section.
\begin{enumerate}[label=\alph*)]
    \item \textbf{Figure 2a} - \textit{Percentage of converged runs after each epoch for different initial thresholds}
    \begin{itemize}
        \item Each epoch comprised of 75 timesteps. 
        \item True Parameters: $H=20, \lambda_o=0.2, \lambda_i=0.55, C_o=0, C_i=25, C_c=50000, \gamma=0.9$. The optimal threshold is 12.
        \item Priors: $\lambda_{o,p}=0.15,\lambda_{i,p}=0.6, C_{o,p}=0, C_{i,p}\in\{0.005,0.1,0.5,2.0\}$. The initial thresholds are $20,18,16,14$ for $C_{i,p}=0.005,0.1,0.5,2.0$, respectively.
        \item Prior Strength: $n_0=1000$
        \item Probability Sequence: $\rho_k=0.95k^{-1/3}$, where $k$ denotes the epoch number.
    \end{itemize}
    \item \textbf{Figure 2b} - \textit{Percentage of converged runs after each epoch for different prior strengths}
    \begin{itemize}
        \item Each epoch comprised of 75 timesteps. 
        \item True Parameters: $H=20, \lambda_o=0.2, \lambda_i=0.55, C_o=0, C_i=25, C_c=50000, \gamma=0.9$. The optimal threshold is 12.
        \item Priors: $\lambda_{o,p}=0.15,\lambda_{i,p}=0.6, C_{o,p}=0, C_{i,p}=0.5$. The initial threshold is 16.
        \item Prior Strength: $n_0\in\{10,100,200,500,1000\}$.
        \item Probability Sequence: $\rho_k=0.95k^{-1/3}$, where $k$ denotes the epoch number.
    \end{itemize}
    \item \textbf{Figure 3a} - \textit{Average time to reach critical health for different initial thresholds}
    \begin{itemize}
        \item Each epoch comprised of 75 timesteps. 
        \item True Parameters: $H=20, \lambda_o=0.2, \lambda_i=0.55, C_o=0, C_i=25, C_c=50000, \gamma=0.9$. The optimal threshold is 12.
        \item Priors: $\lambda_{o,p}=0.15,\lambda_{i,p}=0.6, C_{o,p}=0, C_{i,p}\in\{0.005,0.1,0.5,2.0\}$. The initial thresholds are $20,18,16,14$ for $C_{i,p}=0.005,0.1,0.5,2.0$, respectively.
        \item Prior Strength: $n_0=1000$
        \item Probability Sequence: $\rho_k=0.95k^{-1/3}$, where $k$ denotes the epoch number.
    \end{itemize}
    \item \textbf{Figure 3b} - \textit{Average time to reach critical health for different prior strengths}
    \begin{itemize}
        \item Each epoch comprised of 75 timesteps. 
        \item True Parameters: $H=20, \lambda_o=0.2, \lambda_i=0.55, C_o=0, C_i=25, C_c=50000, \gamma=0.9$. The optimal threshold is 12.
        \item Priors: $\lambda_{o,p}=0.15,\lambda_{i,p}=0.6, C_{o,p}=0, C_{i,p}=0.5$. The initial threshold is 16.
        \item Prior Strength: $n_0\in\{10,100,200,500,1000\}$.
        \item Probability Sequence: $\rho_k=0.95k^{-1/3}$, where $k$ denotes the epoch number.
    \end{itemize}
    \item \textbf{Table 1} - \textit{Safety of deployed policy for different priors}
    \begin{itemize}
        \item Each epoch comprised of 75 timesteps. 
        \item True Parameters: $H=20, \lambda_o=0.2, \lambda_i=0.55, C_o=0, C_i=25, C_c=50000, \gamma=0.9$. The optimal threshold is 12.
        \item Priors: $\lambda_{o,p}=0.15,\lambda_{i,p}=0.6, C_{o,p}=0, C_{i,p}=0.5$. The initial threshold is 16.
        \item Prior Strength: $n_0\in\{10,100,200,500,1000\}$.
        \item Probability Sequence: $\rho_k=0.95k^{-1/3}$, where $k$ denotes the epoch number.
    \end{itemize}
\end{enumerate}
\end{document}